\documentclass{svmult}
  
\usepackage{mathptmx,hyperref}   
\usepackage{helvet} 
\usepackage{courier}
\usepackage{amsmath,amsfonts}
\usepackage{graphicx}
\usepackage[bottom]{footmisc}          
    
\usepackage{mathptmx}         
\usepackage{helvet}            
\usepackage{courier}            
\usepackage{type1cm}             
\usepackage{makeidx}             
\usepackage{graphicx}          
\usepackage{multicol}          
\usepackage[bottom]{footmisc}  
      
\usepackage{epsfig}     
\usepackage{psfrag,rotating}     
\usepackage{amssymb,floatflt,enumerate} 
\usepackage{amsmath,amscd,psfrag,leqno} 
\usepackage[mathscr]{eucal}  
\usepackage{shadethm}           
\usepackage{graphicx,subfigure}   
\usepackage{overpic,contour}       
\contourlength{0.3mm}  
 
\def\Beweisende{\square}            
\def\BewEnde{\hfill{\Beweisende}}

\def\RR{{\mathbb R}}

\def\Vkt#1{{\mathbf #1}}

\newcommand{\go}[1]{{\sf #1}}

\makeindex

\begin{document}

\title*{Isometrically deformable cones and cylinders carrying planar curves}

\author{Georg Nawratil}
\institute{
	Institute of Discrete Mathematics and Geometry \&  
	Center for Geometry and Computational Design, TU Wien 
}

%
%
\maketitle

\abstract{We study cones and cylinders with a 1-parametric isometric deformation carrying at least two planar curves, 
which remain planar during this continuous flexion and are located in non-parallel planes. 
We investigate this geometric/kinematic problem in the smooth and the discrete setting, as it 
is the base for a generalized construction of so-called T-hedral zipper tubes.  
In contrast to the cylindrical case, which can be solved easily, the conical one is more tricky, but we succeed to give 
a closed form solution for the discrete case, which is used to prove that these cones correspond to caps of Bricard octahedra of the 
plane-symmetric type. 
For the smooth case we are able to reduce the problem by means of symbolic computation to an 
ordinary differential equation, but its solution remains an open problem.
}

\section{Motivation and introduction}\label{sec:intro}

The work is motivated by the study of so-called {\it zipper tubes}, which were introduced by Filipov, Tachi and Paulino \cite{pnas}. 
This term always refer to a pair of discrete tubes (in most examples both tubes are congruent) consisting of planar quads where each vertex has valence four. 
Each of these tubes has a one-parametric isometric deformation (also known as rigid-foldability). Moreover, the tubes are glued together along a row of faces, the 
so-called {\it zip row}, and the resulting zipper structure is still continuous flexible. 

In \cite{KMN1} it is realized that most of the known zipper structures consist of tubes, which  
belong to a special class of rigid-foldable quad-surfaces known as T-hedra (cf.\ \cite{graf,sauer,SNRT2021,IRT2022}). 
This perspective leads to a deeper geometric understanding of the underlying zipper principle, which also made it possible to contemplate the construction of 
generalized T-hedral zipper tubes not only in a discrete setting but in a smooth one as well. 
Moreover, in \cite{KMN1} it is shown that this generalized construction reduces to the solution of the following
geometric/kinematic problem studied in this \medskip paper. 

\noindent
{\bf Geometric/kinematic Problem.} {\it Determine all discrete/smooth cones $\Lambda$ and cylinders $\Gamma$, respectively, allowing a 1-parametric isometric deformation $\iota$ in a way 
that at least two planar curves $\go a$ and $\go b$ exist on $\Lambda$  and $\Gamma$, respectively, which remain planar under $\iota$. 
Moreover, we assume that the two carrier planes $\alpha$ and $\beta$ of $\go a$ and $\go b$, respectively, are not \medskip parallel (cf.\ Fig.\ \ref{fig1}a).}

\begin{figure}[t]
\begin{center}
\begin{overpic}
    [height=40mm]{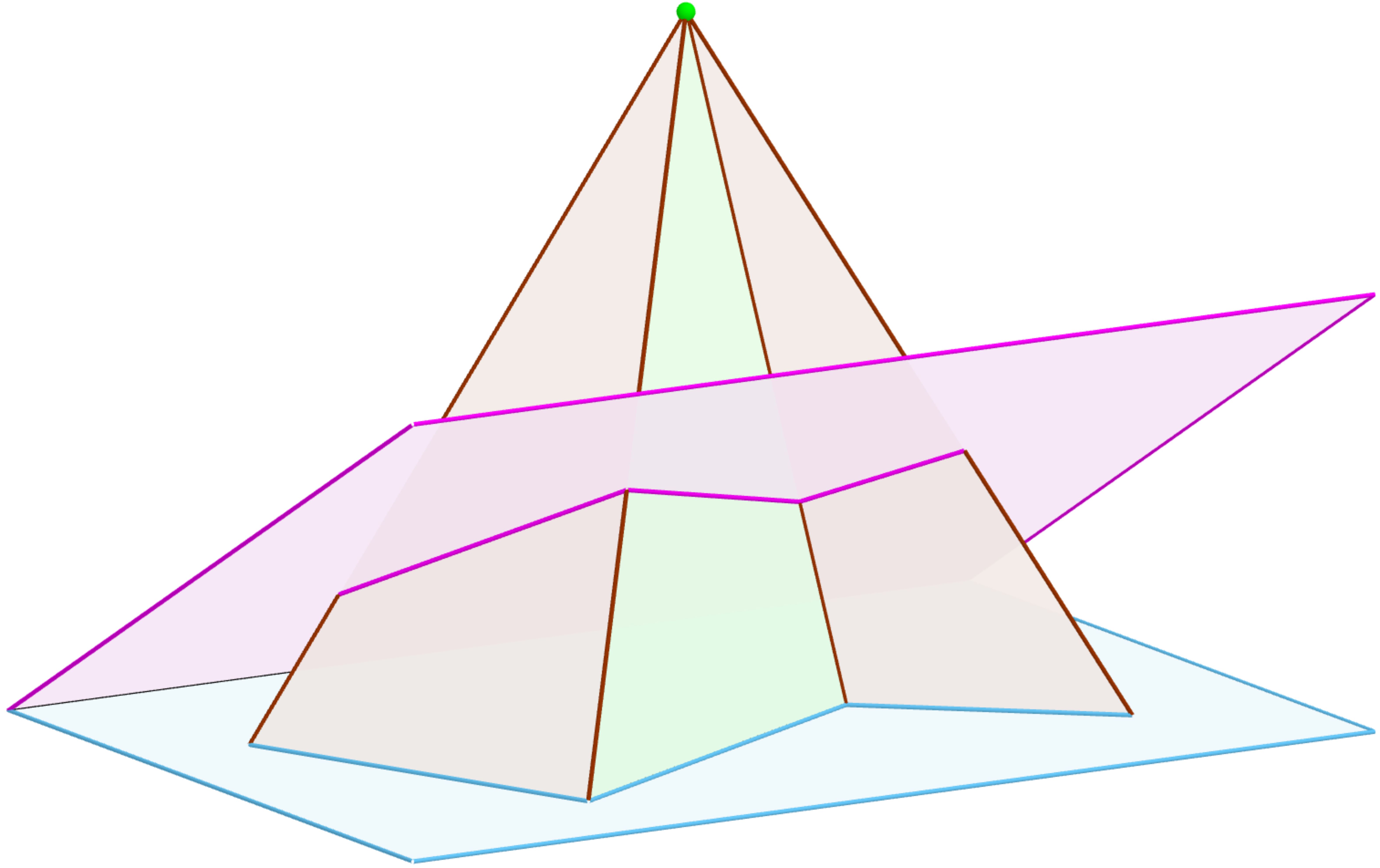}
\begin{scriptsize}
\put(0,0){a)}
\put(89,37){$\alpha$}
\put(88,10){$\beta$}
\put(46,60.5){$V$}
\put(31,8){$b$}
\put(34,21){$a$}
\put(42,45){$f_1$}
\put(49,45){$f_2$}
\put(55,45){$f_3$}

\end{scriptsize}     
  \end{overpic} 
\,
\begin{overpic}
    [height=40mm]{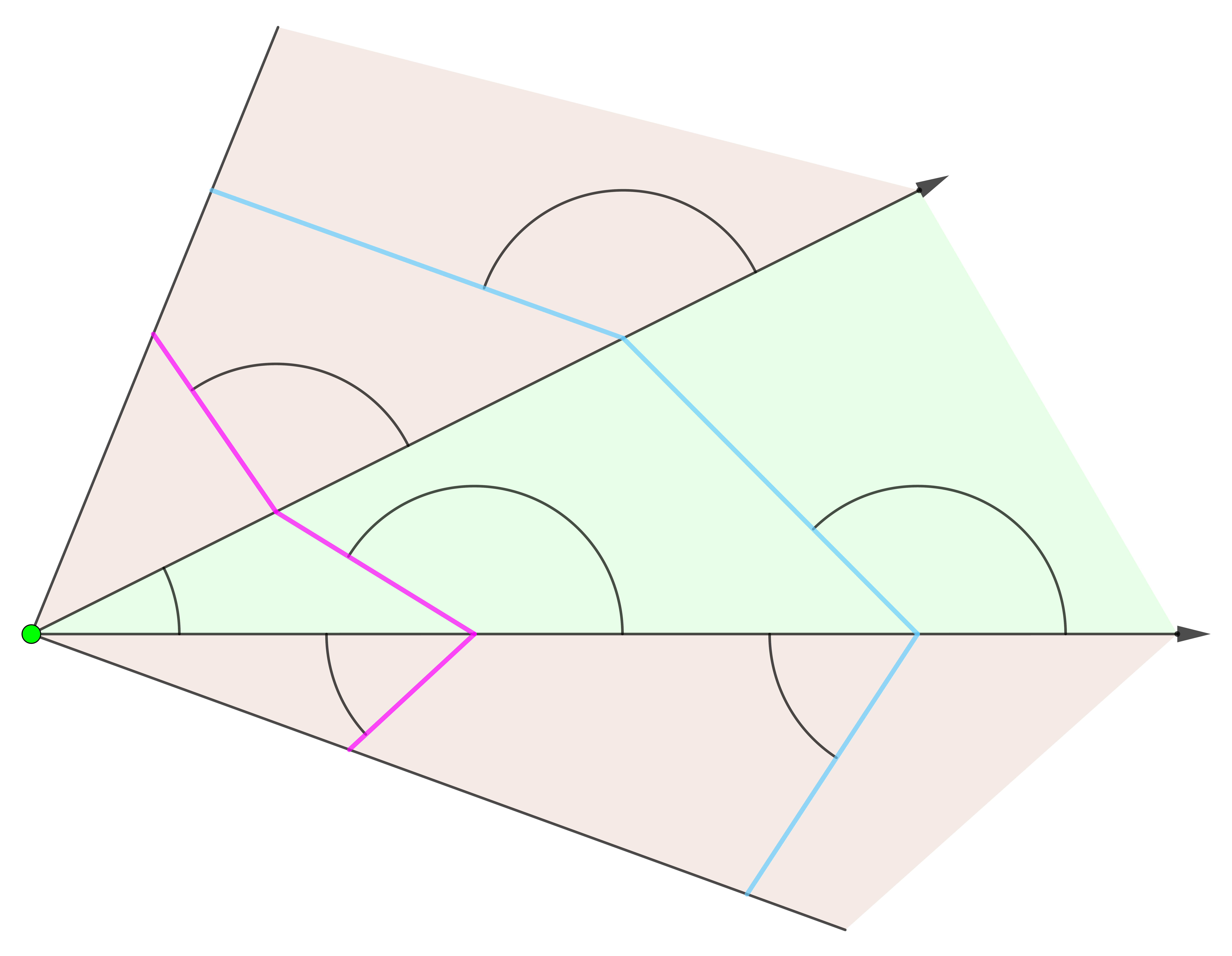}
\begin{scriptsize}
\put(0,0){b)}
\put(34,19){$a_1$}
\put(68,12){$b_1$}
\put(28,22){$\sigma_1$}
\put(65.5,21){$\tau_1$}
\put(0,21){$V$}
\put(38,30){$\sigma_2$}
\put(75,30){$\tau_2$}
\put(80,17){$f_1$}
\put(23.5,30){$a_2$}
\put(10.5,27.5){$\mu$}
\put(78.5,45){$f_2$}
\put(61,41){$b_2$}
\put(95,22){$r_1$}
\put(29,59){$b_3$}
\put(36,67){$f_3$}
\put(48,55){$\tau_3$}
\put(22,41){$\sigma_3$}
\put(13.5,40.5){$a_3$}
\put(77,60){$r_2$}
\end{scriptsize}     
  \end{overpic} 
\end{center}	
\caption{
(a) Projective sketch of the discrete case of the considered geometric/kinematic problem; i.e.\ if $V$ belongs to the ideal plane then the faces 
$f_1,f_2,f_3$ belong to a discrete cylinder $\Gamma$ otherwise to a discrete cone $\Lambda$.
(b) Sketch of the development of the faces $f_1,f_2,f_3$ into the plane.}
  \label{fig1}
\end{figure}
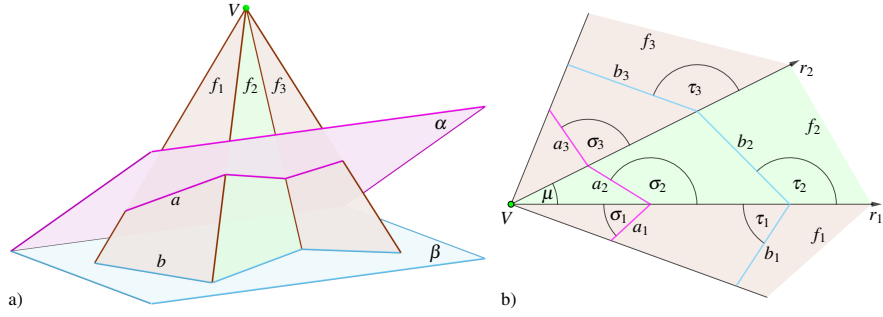

\subsection{Preliminary considerations and outline}\label{sec:preliminary}

The considerations within this subsection hold for the discrete and the smooth case.
Let us assume that we have found a solution to the geometric/kinematic problem stated above, then 
planes parallel to $\alpha$ (resp.\ $\beta$) also carry curves of $\Lambda$ and $\Gamma$, respectively, 
which remain planar during $\iota$. This holds true, as the intersection curve of a parallel plane to $\alpha$ (resp.\ $\beta$) 
results from $\go a$ (resp.\ $\go b$) by a 
\begin{enumerate}
\item
central scaling with center in the vertex  $V$ of $\Lambda$ in the conical case,
\item
translation along the ruling direction of $\Gamma$ in the cylindrical case.
\end{enumerate}
Therefore, we assumed that the planes $\alpha$ and $\beta$ are not parallel.

Moreover, we can produce further planar curves $\go c$ on $\Lambda$ remaining planar under $\iota$, by 
assigning an arbitrary real value to the cross-ratio $(\go a,\go b,\go c,V)$, which has to be evaluated along every ruling through the vertex $V$ of $\Lambda$. 
For the cylindrical case $V$ is an ideal point of $\Gamma$ and the cross-ratio reduces to the ratio $(\go a,\go b,\go c)$. 
Due to the linearity of this construction the obtained curve $\go c$ is again planar and its carrier plane $\gamma$ is contained within the 
pencil of planes spanned by $\alpha$ and $\beta$.
This already shows that there exists at least a one-parametric set of plane orientations implying planar cuts of $\Lambda$ and $\Gamma$, respectively, 
whose planarity is kept under $\iota$.

The above considerations also show that for the cylindrical case one can always assume without loss of generality that  $\beta$
is orthogonal to the ruling direction. As a consequence $\go b$ remains planar under any isometric deformation of $\Gamma$, which 
keeps $\Gamma$ a cylinder through the fixed ideal point $V$.  
Thus the problem gets more relaxed and one just has to compute $\iota$ in a way that $\go a$ remains planar. The corresponding computations 
are straightforward and are given for reasons of completeness in the Appendix (see Section \ref{sec:discrete_cylinder} for the discrete setting 
and Section \ref{sec:smooth_cylinder} for the smooth one). The remainder of the paper is structured as follows: 
The discrete conical case is studied in Section \ref{sec:discrete} and the smooth one in Section \ref{sec:smooth}.  
We conclude the paper in Section \ref{sec:conclusion}.



\section{Discrete conical case}\label{sec:discrete}

For the study of the discrete case it is sufficient to consider three adjacent faces $f_1,f_2,f_3$ of the discretized cone, where the dihedral angle 
enclosed by $f_i$ and $f_{i+1}$ is denoted by $\delta_i$ for $i=1,2$. The corresponding rotation axis is denoted by $r_i$. 
The intersection point of $r_1$ and $r_2$ is the vertex $V$ and the angle enclosed by $r_1$ and $r_2$ is denoted by $\mu$ (cf.\ Fig.\ \ref{fig1}b).

By intersecting two non-parallel planes $\alpha$ and $\beta$ with these three faces, we obtain two planar polygons $a_1,a_2,a_3$ and $b_1,b_2,b_3$, respectively. 
There exists an 1-parametric rigid-folding of the three faces in a way that $a_1,a_2,a_3$ keep their coplanarity. 
Under which geometric conditions also $b_1$, $b_2$ and $b_3$ remain coplanar under this isometric deformation is figured out in the remainder of this section.

\subsection{Setting up the equations}

We consider $f_2$ as fixed and locate it in the $xy$-plane of the reference frame in such a 
way that $V$ equals the origin and that $r_1$ coincides with the $x$-axis. 
Moreover, the direction $A_i$ of the edges $a_i$ can be parametrized as follows:
\begin{equation}
A_3:=(\cos{(\mu+\sigma_3)},\sin{(\mu+\sigma_3)},0)^T, \quad
A_i:=(\cos{\sigma_i},\sin{\sigma_i},0)^T, 
\end{equation} 
where $\sigma_i$ denotes the angle enclosed by $a_i$ and $r_1$ (for $i=1,2$) and $\sigma_3$ denotes the 
angle enclosed by $a_3$ and $r_2$. 
Now we rotate $A_1$ by the angle $\delta_1$ about the $r_1$-axis, which yields 
$A_1^*:=\Vkt R_1 A_1$ with
\begin{equation}
\Vkt R_1:=
\begin{pmatrix}
1 & 0 & 0 \\
0 & \cos{\delta_1} & -\sin{\delta_1} \\
0 & \sin{\delta_1} &  \cos{\delta_1} 
\end{pmatrix}.
\end{equation}
In a similar way we can rotate $A_3$ by the angle $\delta_2$ about the $r_2$-axis, which yields $A_3^*:=\Vkt R_2 A_3$ with
\begin{equation}
\Vkt R_2:=
\begin{pmatrix}
(1-\cos{\mu}^2)\cos{\delta_2}+\cos{\mu}^2    & \cos{\mu}\sin{\mu}(1-\cos{\delta_2}) & \sin{\mu}\sin{\delta_2}\\
\cos{\mu}\sin{\mu}(1-\cos{\delta_2}) & (1-\sin{\mu}^2)\cos{\delta_2}+\sin{\mu}^2 & -\cos{\mu}\sin{\delta_2} \\
-\sin{\mu}\sin{\delta_2} & \cos{\mu}\sin{\delta_2} &  \cos{\delta_2} 
\end{pmatrix}.
\end{equation}
Now the coplanarity of the edges $a_1,a_2,a_3$ can be expressed by $D_1=0$ with
\begin{equation}\label{con1}
D_1:=\det(A_1^*,A_2,A_3^*).
\end{equation}
Exactly the same procedure can be done with respect to the directions $B_i$ of the edges $b_i$, which are given by:
\begin{equation}
B_3:=(\cos{(\mu+\tau_3)},\sin{(\mu+\tau_3)},0)^T, \quad
B_i:=(\cos{\tau_i},\sin{\tau_i},0)^T, 
\end{equation} 
for $i=1,2$.
Then the coplanarity of the edges $b_1,b_2,b_3$ equals the condition $D_2=0$ with
\begin{equation}\label{con2}
D_2=\det(B_1^*,B_2,B_3^*)
\end{equation}
and $B_1^*:=\Vkt R_1 B_1$ and $B_3^*:=\Vkt R_2 B_3$. 

In order to convert the conditions of Eqs.\ (\ref{con1}) and (\ref{con2}) into algebraic ones, we use the half-angle substitution; i.e.\
\begin{equation}
d_i:=\tan{\tfrac{\delta_i}{2}}, \quad s_j:=\tan{\tfrac{\sigma_j}{2}},\quad t_j:=\tan{\tfrac{\tau_j}{2}}, \quad  m:=\tan{\tfrac{\mu}{2}}
\end{equation}
for $i=1,2$ and $j=1,2,3$.


\subsection{Excluded cases}\label{sec:special}

It is not allowed that 
\begin{enumerate}[$\bullet$]
\item
$A_1$ (resp.\ $B_1$) is in direction of $r_1$,
\item
$A_3$ (resp.\ $B_3$) is in direction of $r_2$,
\item
$A_2$ (resp.\ $B_2$) is in direction of $r_1$ or $r_2$,
\end{enumerate} 
which imply $s_1s_3s_2(s_2-m)\neq 0$ (resp.\ $t_1t_3t_2(t_2-m)\neq 0$). 
But also the opposite directions are not allowed which results in the condition 
$\tfrac{ms_2+1}{s_1s_3s_2}\neq 0$ (resp.\ $\tfrac{mt_2+1}{t_1t_3t_2}\neq 0$).

Moreover, a special case arise if $A_1$ and $A_2$ (resp.\  $B_1$ and $B_2$) are identical up to orientations. A necessary condition for that 
is that the carrier planes $f_1$ and $f_2$ of the edges $a_1$ and $a_2$ (resp.\ $b_1$ and $b_2$) coincide which is the case for 
\begin{enumerate}
\item
$\delta_1=0$ ($\Leftrightarrow$ $d_1=0$): 
Then either $\sigma_1=\sigma_2$ (resp.\ $\tau_1=\tau_2$) has to hold which equals the condition $s_1-s_2=0$ (resp.\ $t_1-t_2=0$) 
or $\sigma_1=\sigma_2+\pi$ (resp.\ $\tau_1=\tau_2+\pi$) has to hold which is equivalent to $s_1s_2+1=0$ (resp.\ $t_1t_2+1=0$).
\item
$\delta_1=\pi$ ($\Leftrightarrow$ $d_1=\infty$):
Then either $\sigma_1=-\sigma_2$ (resp.\ $\tau_1=-\tau_2$) has to hold which equals the condition $s_1+s_2=0$ (resp.\ $t_1+t_2=0$) 
or $\sigma_1=-\sigma_2+\pi$ (resp.\ $\tau_1=-\tau_2+\pi$) has to hold which is equivalent to $s_1s_2-1=0$ (resp.\ $t_1t_2-1=0$).
\end{enumerate}

The same considerations can be done for the case that $A_2$ and $A_3$ (resp.\  $B_2$ and $B_3$) are identical up to orientations. A necessary condition for that 
is that the carrier planes $f_2$ and $f_3$ of the edges $a_2$ and $a_3$ (resp.\ $b_2$ and $b_3$) coincide which is the case for $\delta_2=0$ or $\delta_2=\pi$. 
Similar two cases as above can be considered, but there is no need of an explicit discussion for the remaining study.


\subsection{Elimination steps}

From the two conditions $D_1=0$ and $D_2=0$ we can eliminate $d_1$ by means of resultant and obtain a condition of the form
$E_4d_2^4+E_2d_2^2+E_0=0$ where $E_0,E_2,E_4$ are functions in $m,s_j,t_j$ with $j=1,2,3$. 
In order that for all $d_2\in\RR$ a value for $d_1$ exists such that $D_1=0$ and $D_2=0$ hold, the coefficients with respect to $d_2$ have to vanish, which implies 
the  three algebraic equations:
\begin{equation}
E_4=0, \quad E_2=0, \quad E_0=0. 
\end{equation}
In a first step we eliminate $s_2$ by means of resultant, which can be done in three ways:
\begin{equation}
F_0:=Res(E_2,E_4,s_2), \quad 
F_2:=Res(E_0,E_4,s_2), \quad 
F_4:=Res(E_0,E_2,s_2).
\end{equation} 
These three expressions have a greatest common divisor (gcd) of the form 
\begin{equation}\label{eq:G}
G:= 2^8 s_1^8 s_3^8 t_3^4 (t_1-t_2)^2(t_1t_2+1)^2(t_1+t_2)^2(t_1t_2-1)^2.
\end{equation}
For the expressions in the brackets one can easily check by back-substitution that they imply 
exactly the special cases mentioned in item 1 and item 2 of Sec.\ \ref{sec:special}. 
Therefore we can only consider the expressions $G_k:=F_k/G$ for $k=0,2,4$. 

Based on these three expressions we can eliminate again by means of resultant the variable $s_1$; i.e.
\begin{equation}
H_0:=Res(G_2,G_4,s_1), \quad 
H_2:=Res(G_0,G_4,s_1), \quad 
H_4:=Res(G_0,G_2,s_1).
\end{equation} 
Then the solution has to be contained in the gcd of $H_0,H_2,H_4$, which equals:
\begin{equation}
\begin{split}
2^{256}m^{128}(m^2+1)^{32}t_2^{128}t_3^{48}(t_1-t_2)^8(t_1t_2+1)^8(t_1+t_2)^8(t_1t_2-1)^8 \\
(t_2-m)^{128}(mt_2+1)^{128}(s_3-t_3)^{128}(s_3t_3+1)^{128}T_1^{16}T_2^{16}
\end{split}
\end{equation}
with
\begin{align}
T_1&:=t_2(t_1 + t_3)(t_1t_3 + 1)m^2 - t_1(t_3^2 + 1)(t_2^2 - 1)m - t_2(t_1-t_3)(t_1t_3 - 1), \\
T_2&:=t_2(t_1-t_3)(t_1t_3 - 1)m^2 + t_1(t_3^2 + 1)(t_2^2 - 1)m + t_2(t_1 + t_3)(t_1t_3 + 1).
\end{align}
We only have to discuss the last four factors in detail as $m$ has to be a real number different from zero and the other factor
were already identified to belong to special cases. 
The factors $(s_3-t_3)$ and $(s_3t_3+1)$ belong to the trivial solution that the planes $\alpha$ and $\beta$ are parallel which can be seen as follows:
Back-substitution  of $s_3=t_3$ and $s_3=-1/t_3$, respectively,  into $G_0,G_2,G_4$ shows that the resulting expressions can only vanish for 
\begin{equation}
t_2^8t_3^8(s_1-t_1)^8(s_1t_1+1)^8(m^2+1)^8=0
\end{equation}
Back-substitution  of both possible solutions $s_1=t_1$ and $s_1=-1/t_1$, respectively, into $E_0,E_2,E_4$ shows that the resulting expressions can only vanish for 
\begin{equation}
t_1^2t_3^2(s_2-t_2)^2(s_2t_2+1)^2=0. 
\end{equation}
The possible solutions $s_2=t_2$ and $s_2=-1/t_2$, respectively, imply that $\alpha$ and $\beta$ are parallel.

As a consequence only $T_1T_2=0$ can imply a non-trivial solution to our problem. 
Interestingly both conditions are independent of $s_3$; i.e.\ only depend on $m,t_1,t_2,t_3$. This means if this condition is fulfilled then 
one can choose $s_3$ arbitrarily,  which shows that a one-parametric set of planar slices remains planar during the one-parametric rigid-folding of the 
zipper strip (cf.\ Section \ref{sec:preliminary}).

 \subsection{Closed form solution}

The classical method to obtain a closed form solution is to solve $T_1=0$ or $T_2=0$ and substitute the obtained expressions back into $G_0,G_2,G_4$ to compute 
the expression for $s_1$ and substitute the obtained expression again back into $E_0,E_2,E_4$ to get the expression for $s_2$. 
Then the final back-substitution into $D_1=0$ and $D_2=0$ yields the relation between $d_1$ and $d_2$. 
We found a more elegant way, which also yields compact expression, and is based on the following idea: 

As none of the planes $\alpha$ and $\beta$ is privileged the analog condition to $T_1T_2=0$ also has to hold with respect to the $s_i$'s; i.e. $S_1S_2=0$  with 
\begin{align}
S_1&:=s_2(s_1 + s_3)(s_1s_3 + 1)m^2 - s_1(s_3^2 + 1)(s_2^2 - 1)m - s_2(s_1-s_3)(s_1s_3 - 1) \\
S_2&:=s_2(s_1-s_3)(s_1s_3 - 1)m^2 + s_1(s_3^2 + 1)(s_2^2 - 1)m + s_2(s_1 + s_3)(s_1s_3 + 1) \label{eq:ref}
\end{align}
Therefore we obtained two condition $S_1S_2=0$ and $T_1T_2=0$ from the initial system of three equations $E_0=E_2=E_4=0$. 
Thus only one more condition is missing which can be computed in the following way.
We eliminate from $E_i$ ($i=0,2,4$) and $S_u$ ($u=1,2$) the unknown $s_2$ by means of resultant, which yields 
$J_{i,u}:=Res(E_i,S_u,s_2)$. From this expression and $T_v$ ($v=1,2$) we eliminate $t_2$ by computing 
$K_{i,u,v}:=Res(J_{i,u},T_v,t_2)$. 

Then the gcd 
of $K_{0,u,v}$, $K_{2,u,v}$ and $K_{4,u,v}$, which reads as follows:
\begin{equation}
2^8s_1^8s_3^8t_1^8t_3^8m^8(m^2+1)^8M_{u,v}^4N_{u,v}^4,
\end{equation}
has to contain the looked for condition. 
For the four different possibilities the factors $M$ and $N$ are given by:
\begin{align}
M_{1,1}=M_{2,2}&=s_1s_3t_1 - s_1s_3t_3 -s_1t_1t_3 + s_3t_1t_3 - s_1 + s_3 + t_1 - t_3, \label{eq:1} \\
N_{1,1}=N_{2,2}&=s_1s_3t_1 + s_1s_3t_3 -s_1t_1t_3 - s_3t_1t_3 + s_1 + s_3 - t_1 - t_3, \label{eq:2} \\
M_{1,2}=M_{2,1}&=s_1s_3t_1t_3 + s_1s_3 + s_1t_1 - s_1t_3 - t_1s_3 + s_3t_3 + t_1t_3 + 1, \label{eq:3} \\
N_{1,2}=N_{2,1}&=s_1s_3t_1t_3 - s_1s_3 + s_1t_1 + s_1t_3 + t_1s_3 + s_3t_3 - t_1t_3 + 1. \label{eq:4}
\end{align}
%
%
The fulfillment of the conditions $S_1S_2=0$ (resp.\ $T_1T_2=0$) already implies that $D_1$ (resp.\ $D_2$) splits up into several factors. 
This can be seen by computing the resultant of $D_1$ and $S_u$ with respect to $s_2$ (resp.\ $D_2$ and $T_v$ with respect to $t_2$) 
which yields $O_{s,u}:=2^6s_1^2s_3^2(m^2+1)^2P_{s,u}^2Q_{s,u}^2$
(resp.\  $O_{t,v}:=2^6t_1^2t_3^2(m^2+1)^2P_{t,v}^2Q_{t,v}^2$) with
\begin{align}
P_{w,1}&=(w_1w_3+w_1m+w_3m-1)d_1-(w_1w_3-w_1m-w_3m-1)d_2, \label{eq:trans1}\\ 
Q_{w,1}&=(w_1w_3m+w_1-w_3+m)d_1d_2-(w_1w_3m-w_1+w_3+m), \label{eq:trans2}\\ 
P_{w,2}&=(w_1w_3m-w_1-w_3-m)d_1+(w_1w_3m+w_1+w_3-m)d_2, \label{eq:trans3}\\ 
Q_{w,2}&=(w_1w_3-w_1m+w_3m+1)d_1d_2+(w_1w_3+w_1m-w_3m+1), \label{eq:trans4} 
\end{align}
for $w\in\left\{s,t\right\}$.
The common factor of $O_{s,u}$ and $O_{t,v}$ depends if either $M_{u,v}=0$ or $N_{u,v}=0$ holds. 

All in all this yields a closed form solution of the problem at hand. 
One can solve $S_u$ for $s_2$, $T_v$ for $t_2$ and $M_{u,v}$ or $N_{u,v}$ for one of the variables $s_1,s_3,t_1,t_3$. 

Moreover, one gets also the explicit expression of $d_2$ from the common factor of $O_{s,u}$ and $O_{t,v}$. 
Note that the equations (\ref{eq:trans1}--\ref{eq:trans4}) also show that the three faces $f_1,f_2,f_3$ 
possess two flat configurations.

\subsection{Result}

Based on this preparatory work we can prove the following surprising theorem illustrated in Fig. \ref{fig2}a:

\begin{theorem}
The cone $\Lambda$ equals one cap of a Bricard octahedron of the plane-symmetric type. 
\end{theorem}

\begin{proof}
Using the closed form solution one can easily check by direct computations (e.g.\ with Maple), that the lines 
$a_1$ and $a_3$ (resp.\ $b_1$ and $b_3$) are mirror symmetric with respect to the plane $\omega$, which passes through 
the vertex $V$ of the cone $\Lambda$ and is orthogonal to the intersection line of $\alpha$ and $\beta$. 

By means of iteration we see that this property also has to hold true for $a_i$ and $a_{i+2}$ (resp.\ $b_i$ and $b_{i+2}$)
for $i=1,\ldots, n$. This implies that $a_j=a_{j+4}=a_{j+8}=\ldots$ (resp.\  $b_j=b_{j+4}=b_{j+8}=\ldots$) 
for $j=1,\ldots,4$, thus the polygon closes and forms an anti-parallelogram, which is passed multiple times (depending on $n$). 
This finishes already the proof of the theorem.
\hfill$\BewEnde$
\end{proof}

Note that the first face $f_1$ and the last face $f_n$  can be extended in length as no conditions on the opening angles of these two 
faces are implied by the iteration procedure. They can also be rotated by $\pi$ about the lines $r_1$ and $r_{n-1}$, 
respectively\footnote{This corresponds to change of the opening angle's sign.}.

\begin{figure}[t]
\begin{center}
\begin{overpic}
    [height=40mm]{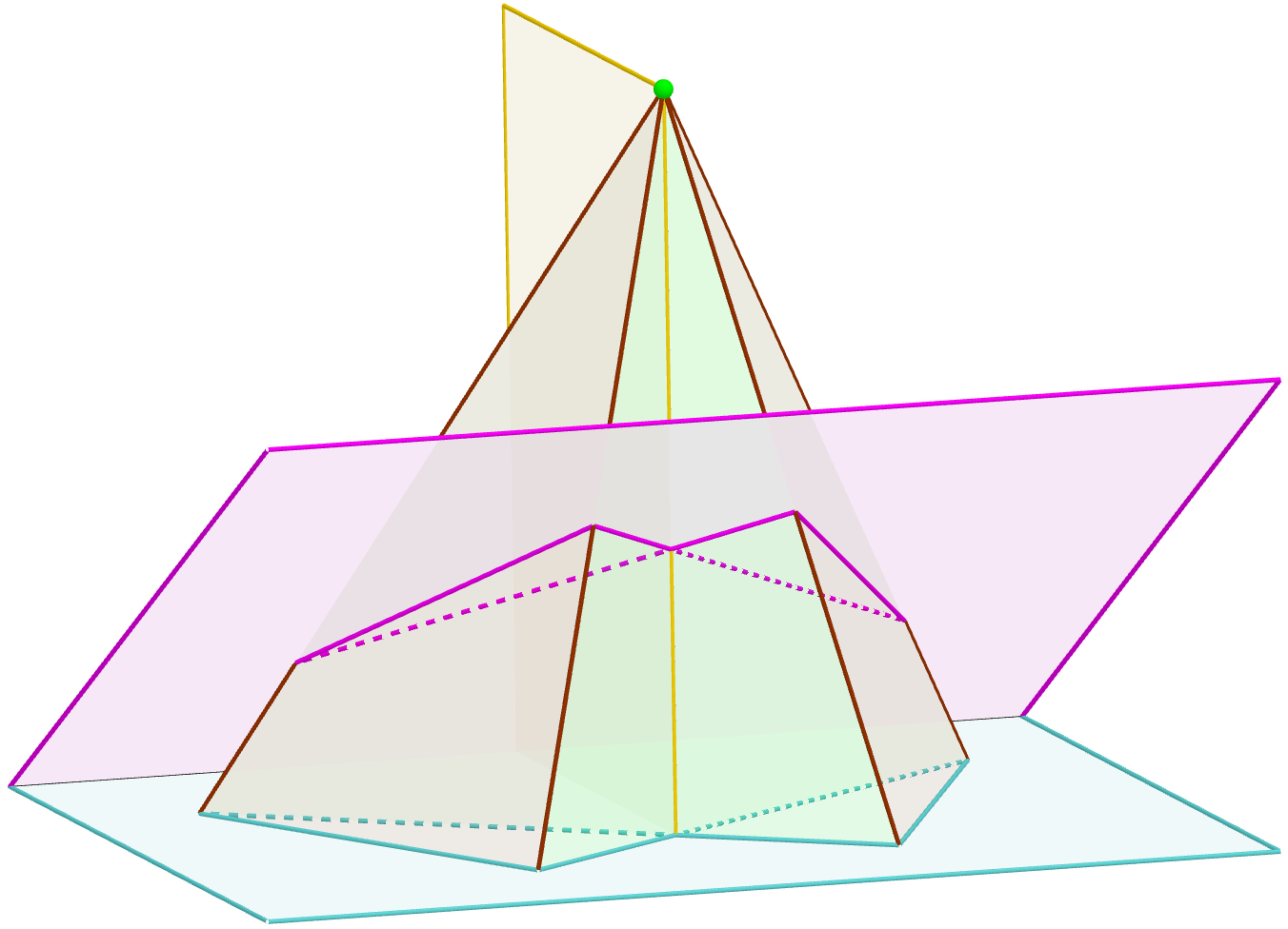}
\begin{scriptsize}
\put(0,0){a)}
\put(53,64){$V$}
\put(40,66.5){$\omega$}
\put(92,38.5){$\alpha$}
\put(88,7){$\beta$}
\put(25.5,3){$b$}
\put(35,28){$a$}
\end{scriptsize}     
  \end{overpic} 
\qquad
\begin{overpic}
    [height=40mm]{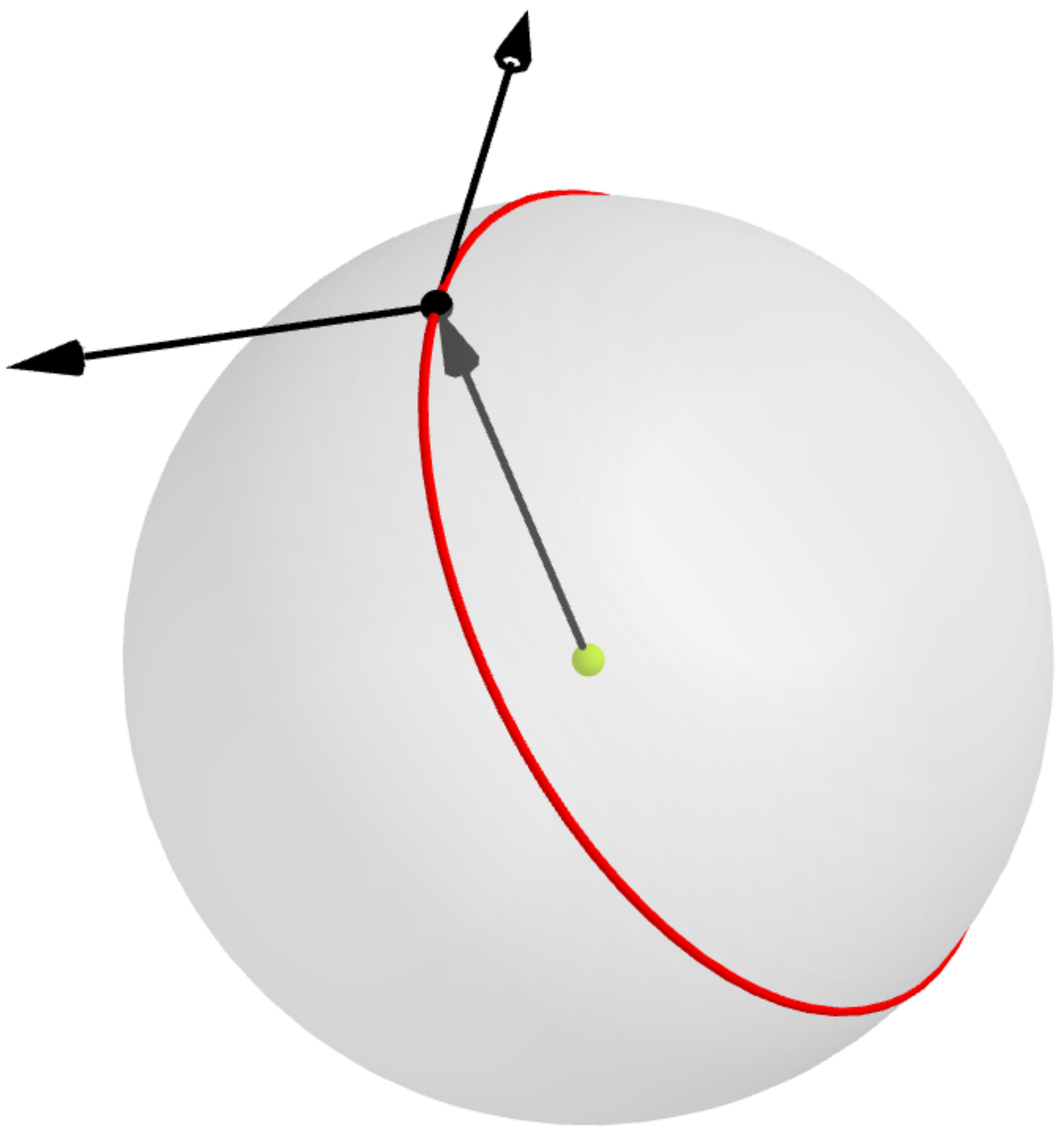}
\begin{scriptsize}
\put(0,0){b)}
\put(54,40){$V$}
\put(40,30){$\go c_{\eta}$}
\put(45.5,59){$\Vkt e_1$}
\put(37,88){$\Vkt e_2$}
\put(11,72){$\Vkt e_3$}
\end{scriptsize}     
  \end{overpic} 
\end{center}	
\caption{
(a) Cap of a Bricard octahedron of the plane-symmetric type.
(b) Sketch of the spherical curve $\go c_{\eta}$ and the Darboux frame $\Vkt e_1, \Vkt e_2,\Vkt e_3$.}
  \label{fig2}
\end{figure}
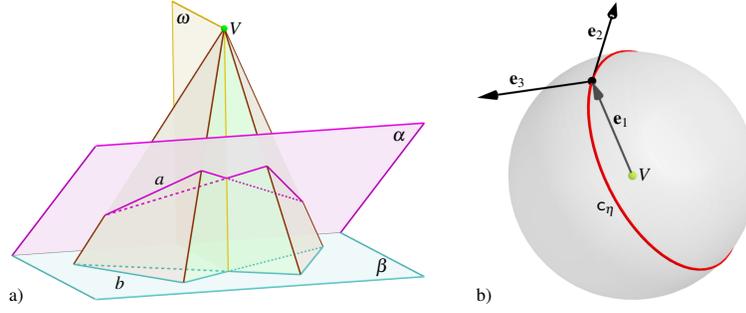

\section{Smooth conical case}\label{sec:smooth}


The discussion of the smooth analog has to based on euclidean differential geometry of curves and cones (e.g.\ \cite{pottmann}). 

We consider an 1-parametric isometric deformable cone $\Lambda_{\eta}$, where $\eta$ denotes the deformation parameter, with vertex in the origin. 
This cone can be parametrized by $\lambda\Vkt e_1(\varsigma,\eta)$ with $\lambda\in\RR$ and $\Vkt e_1$ describes a 1-parametric family of 
spherical curves $\go c_{\eta}$ on the unit sphere; i.e. $\|\Vkt e_1\|=1$. Moreover, due to the isometric deformation we can assume without loss of 
generality that for each curve $\go c_{\eta}$ the parameter $\varsigma$ represents its arc length; i.e.\ $\|\Vkt e_1'\|=1$ where the prime indicates 
the derivative with respect to $\varsigma$. Then one can associate the Cartesian frame $\Vkt e_1$, $\Vkt e_2:=\Vkt e_1'$, and $\Vkt e_3:=\Vkt e_1\times \Vkt e_2$ known as Darboux frame 
with each curve $\go c_{\eta}$. The derivatives of these vectors read as follows \cite[Eq.\ (5.68)]{pottmann}:
\begin{equation}
\Vkt e_1'= \Vkt e_2, \quad \Vkt e_2'=-\Vkt e_1+\kappa\Vkt e_3, \quad \Vkt e_3'=-\kappa\Vkt e_2,
\end{equation}
where $\kappa(\varsigma,\eta)$ is the geodesic curvature of $\go c_{\eta}$. Note that $\kappa$ determines the cone $\Lambda_{\eta}$ uniquely up to Euclidean motions due to the 
fundamental theorem of Euclidean differential geometry of cones \cite[Thm.\ 5.3.13]{pottmann}. 

We can parametrize  curves $\go p_{\eta,i}$ on the cone $\Lambda_{\eta}$ by $\Vkt p_i(\varsigma,\eta):=\phi_i(\varsigma)\Vkt e_1(\varsigma,\eta)$. 
Note, that due to the isometric deformation $\phi_i$ only depends on $\varsigma$ and not on $\eta$.
As $\go p_{\eta,i}$ is not arc length parametrized, its torsion has to be computed according to the well-known formula
\begin{equation}\label{eq:torsion}
\tau_i(\varsigma,\eta):=\frac{(\Vkt p_i'\times \Vkt p_i'')\Vkt p_i'''}{\|\Vkt p_i'\times \Vkt p_i''\|^2}.
\end{equation}
The curve $\go p_{\eta,i}$ is  planar  if and only if $\tau_i(\varsigma,\eta)=0$ holds thus we get as condition 
\begin{equation}\label{eq:must}
(\Vkt p_i'\times \Vkt p_i'')\Vkt p_i'''=0.
\end{equation}
We substitute the expressions 
\begin{equation}
\Vkt p_i'=\phi_i'\Vkt e_1 + \phi_i\Vkt e_1', \quad 
\Vkt p_i''=\phi_i''\Vkt e_1+2\phi_i'\Vkt e_1'+\phi_i\Vkt e_1'', \quad
\Vkt p_i'''=\phi_i'''\Vkt e_1+3\phi_i''\Vkt e_1' +3\phi_i'\Vkt e_1'' +\phi_i\Vkt e_1'''
\end{equation}
under consideration of 
\begin{equation}
\Vkt e_1'=\Vkt e_2, \quad \Vkt e_1''=-\Vkt e_1+\kappa\Vkt e_3, \quad \Vkt e_1'''=\kappa'\Vkt e_3 -(1+\kappa^2)\Vkt e_2
\end{equation}
into Eq.\ (\ref{eq:must}) which yields $K_i=0$ with
\begin{equation}\label{eq:diff_fin}
K_i:=(\phi_i^2 - \phi_i\phi_i'' + 2\phi_i'^2)\phi_i\kappa' 
+ 
(\kappa^2\phi_i^2\phi_i' + \phi_i^2\phi_i' + \phi_i^2\phi_i''' - 6\phi_i\phi_i'\phi_i'' + 6\phi_i'^3)\kappa.
\end{equation}

\begin{remark}
By assuming that the cross-ratio $(\go p_1,\go p_2,\go p_i,V)$ is constant $\lambda_i\in\RR$ 
we can set $\phi_i=\tfrac{\phi_1\phi_2(\lambda_i-1)}{\lambda_i\phi_1-\phi_2}$ in $K_i=0$. 
Then $(\lambda_i-1)^3$ factors away and only an equation of the form $\phi_1^4K_2\lambda-\phi_2^4K_1=0$ remains. 
This shows again that two solutions, which are given by $K_1=0$ and $K_2=0$, imply a 1-parametric set of solutions. 
\hfill $\diamond$
\end{remark}

This differential equation $K_i=0$ 
can be solved (e.g.\ with Maple), 
which yields beside two degenerate solutions (either $\kappa$ or $\phi_i$ is zero)  the following general one:
\begin{equation}\label{eq:kappa}
\kappa=
\frac{
\Phi_i}
{
\phi_i^3\sqrt{W_i}}\quad\text{with}\quad \Phi_i:=\phi_i^2 - \phi_i\phi_i'' + 2\phi_i'^2
\end{equation}
and
\begin{equation}
\begin{split}
W_i:=&
2\int \frac{\phi_i'}{\phi_i\Phi_i} d\varsigma + 
8\int \frac{\phi_i'^3}{\phi_i^3\Phi_i} d\varsigma + 
8\int \frac{\phi_i'^5}{\phi_i^5\Phi_i} d\varsigma -
4\int \frac{\phi_i'\phi_i''}{\phi_i^2\Phi_i} d\varsigma \\
&+ 2\int \frac{\phi_i'\phi_i''^2}{\phi_i^3\Phi_i} d\varsigma 
- 8\int \frac{\phi_i'^3\phi_i''}{\phi_i^4\Phi_i} d\varsigma
+ I(\eta)
\end{split}
\end{equation}
where $I(\eta)$ stands for an arbitrary function in $\eta$. 
As now the function $\kappa$ has to be the same for $i=1,2$ this implies the condition 
$\Phi_2\phi_1^3\sqrt{W_1}=\Phi_1\phi_2^3\sqrt{W_2}$, where we take both sides to the power of two in order to get rid of the square roots; i.e.
\begin{equation}
\Phi_2^2\phi_1^6W_1-\Phi_1^2\phi_2^6W_2=0, 
\end{equation}
which has the structure $U_1(\varsigma)I(\eta) + U_0(\varsigma)=0$. As this condition has to be fulfilled for all values $\eta$ we end up with the 
conditions $U_1(\varsigma)=0$ and $U_0(\varsigma)=0$, which form a system of ordinary differential equations. 
The first condition $U_1(\varsigma)=0$ reads as
\begin{equation}
\Phi_1^2\phi_2^6-\Phi_2^2\phi_1^6=0
\end{equation}
and can easily be solved for $\phi_1$, which yields 
\begin{equation}
\phi_1=\frac{\phi_2}{\phi_2(C_1\sin{\varsigma}-C_2\cos{\varsigma})\pm 1}
\end{equation}
with some constants $C_i\in\RR$. 
If one plugs this expression into $U_0(\varsigma)=0$, one ends up with a final differential equation. 
Unfortunately, we are not able to solve it.

\section{Conclusion}\label{sec:conclusion}

We studied in the smooth and the discrete setting the  geometric/kinematic problem of cones and cylinders, respectively, possessing 
a 1-parametric isometric deformation which preserves the planarity of two curves, which are located in non-parallel planes. 
The cylindrical case turns out to be simple in contrast to the conical one, which we were able to solve in the discrete 
setting by showing that the cones correspond to caps of Bricard octahedra of the plane-symmetric type. 
For the smooth case we were able to solve the resulting system of partial 
differential equations symbolically up to a final ordinary differential equation. Its solution remains an open problem, 
but if there exists one then it would imply a continuous flexible semi-discrete suspension \cite{alexandrov}.

\begin{acknowledgement}
The research is supported by grant  F77 (SFB ``Advanced Computational Design'', SP7) of the Austrian Science Fund FWF.  
\end{acknowledgement}

\section{Appendix}

\subsection{Discrete cylindrical case}\label{sec:discrete_cylinder}

The cylindrical case can be studied similarly to the conical one discussed in Section \ref{sec:discrete} but under consideration 
of $\mu=0$ ($\Rightarrow$ $m=0$). Now the resultant of $D_1$ and $D_2$ with respect to $d_1$ is of the form $E_2d_2^2+E_0$. 
Then the resultant of $E_2$ and $E_0$ with respect to $s_2$ implies the following condition 
\begin{equation}
Gt_2^{12} (t_1-t_3)^2(t_1t_3+1)^2(t_1+t_3)^2(t_1t_3-1)^2R_3^8
\end{equation}
with $G$ of Eq.\ (\ref{eq:G}) and
\begin{equation}\label{eq:ri}
R_i:=s_1^2s_it_1t_i^2 - s_1s_i^2t_1^2t_i - s_1^2s_it_1 + s_1s_i^2t_i + s_1t_1^2t_i - s_it_1t_i^2 - s_1t_i + s_it_1.
\end{equation}
We only have to consider the case $R_3=0$ as the remaining factors to the power of two end up in cases where the faces $f_1$ and $f_3$ are parallel (as   
$a_1$ and $a_3$ are parallel).  

The resultants $Res(R_3,E_2,t_3)$ and  $Res(R_3,E_0,t_3)$ yield the same expression; namely
\begin{equation}
 s_3^4t_1^4(s_1-s_3)^2(s_1s_3+1)^2(s_1+s_3)^2(s_1s_3-1)^2R_2^4
\end{equation}
with $R_i$ of Eq.\ (\ref{eq:ri}). Again we only have to consider $R_2=0$ as the remaining factors to the power of two end up in the trivial cases 
of parallel faces $f_1$ and $f_3$ (as $b_1$ and $b_3$ are parallel).  

For example one can solve $R_i$ for $t_i$ which yield for each  $i=2,3$ two solutions ($\Rightarrow$ four combinations). 
Then $D_1$ is contained as a factor in $D_2$ and we are done. 


\subsection{Smooth cylindrical case}\label{sec:smooth_cylinder}

We consider an 1-parametric isometric deformable cylinder $\Gamma_{\eta}$, where $\eta$ denotes again the deformation parameter. 
This cylinder can be parametrized by $\Vkt c(\varsigma,\eta)+\lambda\Vkt e_3$ with $\lambda\in\RR$ and $\Vkt c$ describes a 1-parametric family of 
planar curves $\go c_{\eta}$ in a plane orthogonal to the rulings direction $\Vkt e_3$ with $\|\Vkt e_3\|=1$. 

Moreover, due to the isometric deformation we can assume without loss of 
generality that for each curve $\go c_{\eta}$ the parameter $\varsigma$ represents its arc length; i.e.\ $\|\Vkt c'\|=1$ where the prime indicates 
the derivative with respect to $\varsigma$. Then one can associate the Cartesian frame $\Vkt e_1:=\Vkt c'$, $\Vkt e_2:=\Vkt e_3\times \Vkt e_1$ and $\Vkt e_3$ 
with each curve $\go c_{\eta}$. The derivatives of these vectors read as follows:
\begin{equation}
\Vkt e_1'= \kappa\Vkt e_2, \quad \Vkt e_2'=-\kappa\Vkt e_1, \quad \Vkt e_3'= \Vkt o,
\end{equation}
where $\kappa(\varsigma,\eta)$ is the curvature of the planar curve $\go c_{\eta}$. Note that $\kappa$ determines the cylinder $\Gamma_{\eta}$ uniquely up to Euclidean motions due to the 
fundamental theorem of Euclidean differential geometry of planar curves. 

Moreover, we can parametrize the curves $\go p_{\eta}$ on the cylinder $\Gamma_{\eta}$ by $\Vkt p(\varsigma,\eta):=\Vkt c(\varsigma,\eta) + \phi(\varsigma)\Vkt e_3$. 
Note, that due to the isometric deformation $\phi$ only depends on $\varsigma$ and not on $\eta$ and $\Vkt e_3$ is constant. 
Again the torsion $\tau$ of $\go p_{\eta}$ has to be computed like in Eq.\ (\ref{eq:torsion}) thus the condition for planarity reads as 
$(\Vkt p'\times \Vkt p'')\Vkt p'''=0$. Within this expression we substitute
\begin{equation}
\Vkt p'=\Vkt e_1 + \phi'\Vkt e_3, \quad 
\Vkt p''=\Vkt e_1' + \phi''\Vkt e_3, \quad
\Vkt p'''=\Vkt e_1'' + \phi'''\Vkt e_3
\end{equation}
under consideration of 
\begin{equation}
\Vkt e_1'=\kappa\Vkt e_2, \quad \Vkt e_1''=\kappa'\Vkt e_2-\kappa^2\Vkt e_1,
\end{equation}
which yields:
\begin{equation}\label{check_substi}
\phi'\kappa^3+\phi'''\kappa-\phi''\kappa'=0.
\end{equation}
\begin{remark}
If we substitute $\phi$ by $\lambda\phi$ with $\lambda\in\RR$ in Eq.\ (\ref{check_substi}) then we see that 
we remain with the same equation just multiplied with the factor\footnote{$\lambda=0$ yields the planar curve in the $e_1e_2$-plane.} $\lambda$. 
\hfill $\diamond$
\end{remark}

This differential equation (\ref{check_substi})
can be solved (e.g.\ with Maple), 
which yields the solution
\begin{equation}\label{eq:kappa_planar}
\kappa=\frac{\phi''}{\sqrt{I(\eta)-\phi'^2}}
\end{equation}
where $I(\eta)$ denotes again an arbitrary function in $\eta$. 

\end{document}